\documentclass[12pt]{amsart}
\usepackage{xcolor}
\usepackage{soul, url}
\usepackage{amsmath, amsfonts, amssymb, amsthm}
\usepackage{hyperref, comment}
\usepackage{graphicx} 
\usepackage[margin=1.5in]{geometry}
\usepackage{amsmath, amsthm, latexsym, amssymb, amsfonts, epsf, xcolor, hyperref, mathrsfs}

\usepackage{booktabs}

\usepackage[nameinlink]{cleveref}
\usepackage{longtable}
\usepackage[table]{xcolor}
\newcommand{\F}{\mathbb{F}}
\newcommand{\N}{\mathbb{N}}
\newcommand{\Z}{\mathbb{Z}}

\newcommand{\x}{\mathbf{x}}
\newcommand{\y}{\mathbf{y}}
\newcommand{\w}{\mathbf{w}}
\newcommand{\uv}{\mathbf{u}}
\newcommand{\vv}{\mathbf{v}}
\newcommand{\cv}{\mathbf{c}}
\DeclareMathOperator{\supp}{supp}
\newcommand{\rmv}[1]{}

\usepackage{xcolor}

\usepackage{url}
\usepackage{wrapfig}
\usepackage{float}
\usepackage{geometry}
\usepackage{amssymb}
\usepackage{amsfonts}
\usepackage{amsthm}
\usepackage{mathtools}
\usepackage{enumerate}
\usepackage{verbatim}
\usepackage{physics}
\usepackage{euscript}
\usepackage{tabularx}
\usepackage[ruled,vlined,linesnumbered]{algorithm2e}
\usepackage[utf8]{inputenc}

\numberwithin{equation}{section}

\usepackage{enumitem}
\usepackage{tikz,pgfplots}
\usepgfplotslibrary{groupplots}
\pgfplotsset{compat=1.18}

\newtheorem{theorem}{Theorem}[section]

\newtheorem{proposition}[theorem]{Proposition}
\newtheorem{corollary}[theorem]{Corollary}

\newtheorem{example}[theorem]{Example}

\theoremstyle{definition}
\newtheorem{definition}[theorem]{Definition} 
\newtheorem{remark}[theorem]{Remark}

\title[COFI-DQI: Curve-based Optimal Function Intersection via DQI]{COFI-DQI: Curve-based Optimal Function Intersection via Decoded Quantum Interferometry}

\author[G. L. Matthews]{Gretchen L. Matthews}
\address[Gretchen L. Matthews]{Department of Mathematics\\ Virginia Tech\\ 
Blacksburg, VA, USA}
\email{gmatthews@vt.edu}

\author[J. Shapiro]{Julia Shapiro}
\address[Julia Shapiro]{Department of Mathematics\\ Virginia Tech\\ 
Blacksburg, VA, USA}
\email{juliams22@vt.edu}

\thanks{Gretchen L. Matthews was partially supported by NSF grants DMS-2502705 and CNS-2413218 and the Commonwealth Cyber Initiative. Julia Shapiro is supported by the Department of War Cyber Service Academy Scholarship}


\begin{document}

\begin{abstract} In 2025, Jordan et al. introduced Decoded Quantum Interferometry (DQI), a quantum algorithm for combinatorial optimization based on decoding. 
They considered Reed--Solomon decoding and its associated optimization problem, called Optimal Polynomial Intersection (OPI), which may be viewed as a polynomial regression problem over a finite field. DQI exhibits provable speedups on certain problem instances and establishes a connection between decoding problems and optimization tasks. Leveraging the well-understood dual structure and decoding theory of algebraic geometry codes from other curve families, we introduce COFI: Curve-based Optimal Function Intersection. By considering two-point Hermitian codes, Suzuki codes, and extended norm--trace codes, we broaden the range of algebraic geometry codes used in DQI and identify families that offer further improvements over one-point Hermitian codes in the Hermitian Optimal Polynomial Intersection framework considered by Jordan and Gu. Depending on the family and parameter regime, these curves can reduce quantum resource requirements or increase the number of constraints that can be considered.  
\end{abstract}

\maketitle

\section{Introduction}
Given a linear system of equations, the Maximum Linear Satisfiability (max-LINSAT) problem is  to find an assignment to the variables that satisfies the maximum possible number of equations. Max-LINSAT is a maximum feasible subsystem problem, as are robust regression and consensus maximization, where the objective is to satisfy as many constraints or data points as possible despite noise, inconsistencies, or outliers. Such formulations arise in operations research, machine learning, computer vision, and signal processing, where optimization is often driven by maximizing agreement with observed data. In this paper, we study max-LINSAT over the finite field \(\F_q\) and the speedups offered by Decoded Quantum Interferometry on certain structured instances.

Consider a system of $n$ linear equations  
\begin{equation} \label{eq:system}
\begin{array}{ccc}
b_{11}x_1+\dots+b_{1k}x_k&=&c_1\\&\vdots & \\
b_{n1}x_1+\dots+b_{nk}x_k&=&c_n
\end{array}
\end{equation}
in $k$ unknowns $x_1, \dots, x_k$ over $\F_q$; here, $b_{ij}, c_i \in \F_q$ for all $i \in [n]:=\{ 1, \dots, n\}$ and $j\in [k]$. The max-LINSAT problem over $\F_q$ is a generalized linear constraint satisfaction problem that seeks to maximize the number of equations satisfied given $r$ acceptable values $c_{i1}, \dots, c_{ir} \in \F_q$ for each $i \in [n]$.  Max‑LINSAT is a finite‑field specialization of the maximum feasible subsystem problem introduced by Amaldi and Kann  \cite{AMALDI1995181}. Using PCP techniques, Hastad later showed that even for linear equations modulo $2$, meaning $q=2$, no polynomial‑time algorithm can approximate the optimum beyond the trivial random‑assignment baseline, under the standard complexity assumptions stated in \cite{hastad2001optimal}. A uniformly random assignment $\x \sim \operatorname{Unif}(\F_q^k)$ satisfies 
\(b_{i1}x_1+\dots+b_{ik}x_k=c_i\)
for a particular $i \in [n]$ with probability $\frac{r}{q}$, for each $i \in [n]$, assuming that each row $(b_{i1},\ldots,b_{ik})$ is nonzero and $\mid \{ c_{i1}, \dots, c_{ir} \}\mid=r$. 
By linearity of expectation, the expected number of satisfied constraints for this uniformly random assignment is $\frac{rn}{q}$. 

Given the computational difficulty of max-LINSAT and related constraint satisfaction problems, identifying structured instances for which improved optimization guarantees are possible is of considerable interest. Recent work has shown that algebraic structure can be exploited in a quantum setting, leading to new approaches that connect optimization problems to decoding problems from coding theory.
In particular, a major advance was Decoded Quantum Interferometry (DQI), introduced by Jordan, Shutty, Wootters, Zalcman, 
Schmidhuber, King,
Isakov, Khattar, and Babbush \cite{jordan2025dqi_nature}; see also \cite{Jordan_2025_extended}. DQI is a quantum algorithm that achieves superpolynomial speedups on certain structured optimization problems
by reducing them to decoding problems for related classical error-correcting codes. 
The algebraic structure underlying some optimization tasks can be leveraged,  allowing
quantum interference to amplify solutions that correspond to codewords
nearest to a target point. The recent exposition by Google Quantum AI \cite{JordanShutty2025QuantumToolkit} describes DQI as a quantum link between optimization and decoding, the process of determining a codeword from a perturbed version, in coding
theory. We note that when $q=2$, the problem is known as max-XORSAT.  

DQI uses a quantum Fourier transform to convert a linear optimization problem, such as max-LINSAT, into a structured decoding task. At a high level, a quantum algorithm first prepares a superposition encoding an objective function, which in the case of max-LINSAT tracks the number of satisfied constraints. Next, it interferes computational paths via quantum Fourier techniques.  Then it calls a decoder as an oracle inside the quantum routine. A key feature of DQI is a reduction that transforms a max-LINSAT instance into a syndrome-decoding problem for the dual of a related linear code. The authors demonstrate that for certain algebraically structured families, including those related to Reed--Solomon codes, the reduction yields an efficiently decodable code and DQI achieves provable speedups. This breakthrough work was followed by Jordan and Gu \cite{Jordan_Gu_AG_DQI_25} in the algebraic setting, who showed that Hermitian codes require approximately $\frac{1}{3}$ fewer qubits per field element than Reed--Solomon codes of the same length for the quantum implementations considered there.

In this paper, we leverage codes from algebraic curves of positive genus, which enable handling a larger number of constraints for a fixed field size compared to previous constructions. A key advantage of the codes we employ is that their duals remain within the same geometric family, making them particularly well-suited for the DQI framework. Moreover, efficient general algebraic geometry (AG) code decoding algorithms can be applied to the dual codes required in our setting. Efficient decoding algorithms are known for broad classes of algebraic geometry codes. For one-point AG codes, there are syndrome-based methods generalizing the Berlekamp--Massey algorithm \cite{SAKATA1990207,476248} and the Feng--Rao algorithm \cite{179340} that provide unique decoding. Later extensions provide efficient unique decoding algorithms for multipoint and dual multipoint AG codes \cite{6744599,6762981,8294201}. 
List decoding was subsequently developed by Guruswami and Sudan \cite{Guruswami1998ImprovedDO}. More recently, efficient interpolation-based list-decoding algorithms have been developed for general algebraic geometry codes, including fast Guruswami--Sudan-style algorithms and subsequent improvements in both computational complexity and decoding radius \cite{beelen2024list, beelen2025faster,Beelen2022FastDO}.
By harnessing families of curves with many rational points, we obtain improvements in DQI relative to Reed--Solomon codes of the same length in terms of qubit requirements per field element. Some also improve upon the Hermitian-based construction. 

\subsection*{Contributions.} Our primary contributions are the following.

We introduce COFI, a curve-based generalization of OPI and HOPI. This framework expands the class of constraint-satisfaction problems to which DQI may be applied. COFI builds on Hermitian Optimal Polynomial Intersection (HOPI) \cite{Jordan_Gu_AG_DQI_25} by extending the setting to finding rational functions on a projective curve $\mathcal X$ over $\F_q$ that maximize agreement at points on $\mathcal X$. For some curve families, these rational functions can be expressed as multivariate polynomials. More generally, COFI allows optimization over functions in a Riemann-Roch space that need not admit a polynomial representation. Hence, we make use of Riemann-Roch spaces of divisors of curves over finite fields.   

    We derive DQI satisfaction-fraction formulas for AG codes arising from general curves. Because the number of $\F_q$-rational evaluation points on $\mathcal X$ determines the number of constraints, we focus on explicit curves with many known rational points. 

  We analyze one-point Suzuki and extended norm--trace codes, as well as two-point Hermitian codes, within this framework. We show that Suzuki codes can achieve improved satisfaction fractions while requiring fewer qubits relative to Hermitian codes of comparable length. Under certain assumptions on the divisor degree, we further show that one-point Suzuki and extended norm--trace codes attain larger satisfaction fractions than the corresponding Hermitian codes. We establish that a class of two-point Hermitian codes yields a strictly larger expected DQI satisfaction fraction whenever the comparison lies in the nonsaturated regime; if both instances are saturated, both fractions equal \(1\).

\subsection*{Organization.}

This paper is organized as follows. Section~\ref{sec:prelim} provides the necessary background on linear codes and algebraic geometry codes. Section~\ref{sec:main} presents the general results on the DQI satisfaction fraction, together with examples of max-LINSAT formulations. 
Section~\ref{sec:families} 
specializes these results to particular code families, compares their DQI satisfaction fractions with those obtained from one-point Hermitian codes, and compares Suzuki-code-based DQI with Prange's algorithm.
Finally, Section~\ref{sec:conclusion} concludes the paper.

\section{Preliminaries} \label{sec:prelim}

\subsection{Linear Codes}
We begin with a review of linear codes and their key properties to be used later. Standard references for linear codes include \cite{
HuffmanPless2003, MacWilliamsSloane1977, PlessHuffman1998}. 

Let $q$ be a power of a prime, and $\F_q$ denote the finite field with $q$ elements. The set of $m \times n$ matrices with entries in $\F_q$ is denoted $\F_q^{m \times n}$. We write $\F_q^n$ for the vector space of length-$n$ vectors over $\mathbb F_q$; vectors will be regarded as rows or columns as dictated by context. Given a vector $\vv \in \F_q^n$ and $i \in [n]$, its $i${th} coordinate is denoted $v_i$, its support is $supp(\vv):=\left\{ i \in [n]: v_i \neq 0 \right\}$, and its (Hamming) weight is $\operatorname{wt}_H(\vv):=\mid supp(\vv) \mid$. At times, if the description of $\vv$ is complicated, we may write $(v)_i$ instead of $v_i$. The (Hamming) distance between vectors $\uv, \vv \in \F_q^n$  is \(d_H(\uv,\vv):=\mid \left\{ i \in [n]: u_i \neq v_i \right\} \mid\).  

An $[n,k,d]$ code $C$ over $\F_q$ is a $k$-dimensional subspace of $\F_q^n${ with minimum distance \(d:=\min \left\{ d_H(\cv,\cv'): \cv, \cv' \in C, \cv \neq \cv' \right\}\). Elements of $C$, called codewords, are vectors with $n$ coordinates each of which lies in $\F_q$. The Singleton Bound states that the parameters of an $[n,k,d]$ code satisfy $k+d\leq n+1$. The dual of an $[n,k,d]$ code $C$ over $\F_q$ is 
\[C^{\perp}:=\left\{ \vv \in \F_q^n: \vv \cdot \cv = 0 \ \forall \cv \in C \right\},\] where $\vv \cdot \cv:=\sum_{i\in [n]} v_i c_i$ denotes the Euclidean inner product over $\F_q$. Observe that $C^{\perp}$
is an $[n,n-k,d^{\perp}]$ code over $\F_q$ for some integer $d^{\perp}$. According to the Singleton Bound,  $d^{\perp} \leq k+1$. The elements of $C^{\perp}$ are referred to as dual codewords of $C$. 

\subsection{Algebraic Geometry Codes}

We use 
codes from algebraic geometry as we describe now. References for algebraic geometry codes include \cite{niederreiter2009algebraic, Stichtenoth}. 

Consider a smooth projective curve
$\mathcal X$ over $\F_q$ given by an equation $f(x,y)=0$, equivalently by $F(X,Y,Z)=0$ where $F$ is the homogenization of $f$. The points on $\mathcal X$ can be expressed as $(a:b:c):=\left\{ (\lambda a, \lambda b, \lambda c): \lambda \in \overline{\F_q} \setminus \{ 0 \} \right\}$ satisfying $F(a,b,c)=0$ where $\overline{\F_q}$ denotes the algebraic closure of $\F_q$ and at least one of $a$, $b$, and $c$ is nonzero.  
Every projective point has a representative with $c=1$ if it is affine and with $c=0$ if it lies at infinity.
Hence, each point on $\mathcal X$ can be represented as $(a:b:c)$ with $c \in \{ 0, 1 \}$. With this in mind, we sometimes write $P_{ab}$ to denote the point $(a:b:1)$ and refer to such a point as an affine point. Points $(a:b:0)$ are called points at infinity.  
The set of $\F_q$-rational points of $\mathcal X$, meaning those points $(a:b:c)$ with $a,b,c \in \F_q$,  is denoted by $\mathcal X(\F_q)$. 
The
 divisor group of $\mathcal X$, denoted $Div(\mathcal X)$,
is the
free abelian group on $\mathcal X(\overline{\F_q}),$
and its elements are called divisors. Hence, a divisor $A$ on $\mathcal X$ can be expressed as $A=\sum_{P \in \mathcal X(\overline{\F_q})} a_P P$ for some $a_P \in \Z$ with $a_P=0$ for all but finitely many $P \in \mathcal X(\overline{\F_q})$. Its support is $\supp(A):=\left\{ P \in \mathcal X(\overline{\F_q}): a_P \neq 0 \right\}$. If $\supp(A) \subseteq \mathcal X({\F_q})$, then the degree of $A$ is $\deg(A):=\sum_{P \in \mathcal X({\F_q})} a_P$. 
Let $\F_q(\mathcal X)$ denote the field of rational functions on $\mathcal X$. Each nonzero function $f \in \F_q(\mathcal X)$ has an associated divisor \((f) = \sum_{P \in \mathcal X(\overline{\F_q})} v_P(f) P\) where $v_P(f)>0$ (resp., $v_P(f)<0$) if $f$ has a zero at $P$ (resp., pole at $P$). The Riemann-Roch space of a divisor $A$ on $\mathcal X$ is \(\mathcal L(A):=\left\{ f \in \F_q(\mathcal X): (f) \geq -A \right\} \cup \{ 0 \}\), which is an $\F_q$-vector space of finite dimension, denoted by $\ell(A)$. A divisor $K$ on $\mathcal X$ is said to be canonical if it is the divisor of a differential on $\mathcal X$, equivalently, if $\deg K = 2g-2$ and $\ell(K)=g$. 
The Riemann-Roch Theorem indicates that \[\ell(A)=\deg A + 1 - g + \ell(K-A)\] where $K$ is a canonical divisor. Hence, if $\deg A \geq 2g-1$, then $\ell(K-A)=0$ and $\ell(A)=\deg A + 1 - g$. 

Throughout this paper, $\mathcal X$ denotes a smooth, projective curve over a finite field of genus $g$. When an affine or projective plane equation is given, $\mathcal X$ denotes a nonsingular projective model. 
Algebraic geometry codes were first defined by Goppa \cite{Goppa1981}. They were later used to construct sequences of codes exceeding the Gilbert-Varshamov Bound \cite{TsfasmanVladutZink1982}. Since then, they have been used in numerous settings due to their flexibility and provably good parameters linked to properties of the underlying curves. For more details, see \cite{Stichtenoth,TsfasmanVladut1991}.

\begin{definition}[Algebraic Geometry Code] \label{def:AGcode}
    Let $\mathcal X$ be a smooth projective curve over $\F_q$ of genus $g$ with at least one $\F_q$-rational point. Take
$\mathcal P=\{P_1,\dots,P_n\}\subseteq \mathcal X(\mathbb{F}_q)$  to be a set of $n$ distinct $\F_q$-rational points.
Let $G$ be a divisor on $\mathcal X$ with $\mathrm{supp}(G)\cap \mathcal P =\emptyset$. The associated algebraic geometry code is $C_{\mathcal L}(D,G):=ev(\mathcal L(G)) \subseteq \F_{q}^n$ where $D=P_1+\dots+P_n$ and 
\[
\begin{array}{lccc}
ev: & \mathcal{L}(G) & \rightarrow & \F_{q}^n \\
& f & \mapsto & \left(f(P_1), \dots, f(P_n) \right).
\end{array}
\]
\end{definition}

We note that there is a slight abuse of notation in Definition \ref{def:AGcode} (though it is standard in the literature), given that the divisor $D$ does not force an order on the points in its support while the definition of $C_{\mathcal L}(D,G)$ does. 
It can be shown that $C_{\mathcal L}(D, G)$ has length $n$; dimension $\ell(G)-\ell(G-D)$, which is the difference of dimensions of the Riemann-Roch spaces $\mathcal L(G)$ and $\mathcal L(G-D)$; and minimum distance $d$ at least $n- \deg G$. If $\deg G<n$, then the evaluation map is injective. Then the dimension of $C_{\mathcal L}(D, G)$
is $\ell(G)$ which is precisely $\deg G+1-g$ in the case that $2g-2 < \deg G$. Summarizing, $C_{\mathcal L}(D, G)$ is an $[n,\ell(G)-\ell(G-D), \geq n-\deg G]$ code over $\F_q$.
Given
a basis $\left\{f_1, \dots, f_k \right\}$ for $\mathcal L(G)$, $C_{\mathcal L}(D, G)$ may be described explicitly via a generator matrix
\[ 
\left[
\begin{array}{cccc}
   f_1(P_1)  & f_1(P_2) & \dots & f_1(P_n) \\
   f_2(P_1)  & f_2(P_2) & \dots & f_2(P_n) \\ \vdots & \vdots & & \vdots \\
f_k(P_1)  & f_k(P_2) & \dots & f_k(P_n)  
\end{array}
\right].
\]
If $\mid \supp G \mid=m$ and $\supp D =\mathcal X(\F_q) \setminus \supp G$, then $C_{\mathcal L}(D,G)$ is called an $m$-point code.  It is a well-known fact (see \cite[Proposition 2.2.10]{Stichtenoth}) that the dual of an algebraic geometry code on a curve $\mathcal X$ is an
algebraic geometry code from the same curve $\mathcal X$.  
In particular, 
the dual of  $C_{\mathcal L}(D,G)$ is
\[C_{\mathcal L}(D,G)^{\perp}=
 C_{\mathcal L}(D,D-G+K)
\] where $K=\left( \eta \right)$ is a canonical divisor such that $v_{P_i}(\eta)=-1$ and $\eta_{P_i}(1)=1$ for all $i \in [n]$. According to \cite[Corollary 2.2.11]{Stichtenoth}, if $\deg G > 2g-2$, then $C_{\mathcal L}(D,G)^{\perp}$ is an $[n,n+g-1-\deg G, d^{\perp}]$ code where 
the minimum distance $d^{\perp}$  
satisfies \[\deg G -2g+2 \leq d^{\perp} \leq \deg G -g+2.\]
Another way to define the dual is through differentials. It is well known that $C_{\mathcal{L}}(D,G)$ and $C_{\Omega}(D,G)$ are dual to each other, where $C_{\Omega}(D,G)$ is another algebraic geometry code based on Weil differentials (see \cite{Stichtenoth}). In particular, 
$$C_{\Omega}(D,G) = C_{\mathcal{L}}(D,G)^{\perp}.$$

\begin{definition}[Dual-code decoding radius] \label{def:dual_decoding_radius} Given a linear code $C$, the (unique) decoding radius of its dual $C^\perp$ is \[ \ell=\left\lfloor\frac{d^\perp-1}{2}\right\rfloor, \] 
where $d^{\perp}$ denotes the minimum distance of $C^{\perp}$.
\end{definition}

Note that $\ell$ is the maximum number of errors that can be corrected uniquely using bounded-distance decoding of $C^\perp$.

The most commonly studied algebraic geometry codes are the Reed--Solomon and Hermitian codes. 
A Reed--Solomon code is a one-point code $C_{\mathcal L}(D,\alpha P_{\infty})$ obtained by taking $\mathcal X=\mathbb{P}^1_{\F_q}$, the projective line over $\F_q$. While its parameters satisfy the Singleton Bound, its length is at most $q$ which is considered short relative to the alphabet size of $q$. In the context of DQI, as considered in \cite{Jordan_2025_extended, jordan2025dqi_nature}, this means that the number of constraints is (at most) $q$. It is worth noting that $\mathcal L(\alpha P_{\infty})=\F_q[x]_{ \leq \alpha}$, the set of univariate polynomials in the variable $x$ of degree at most $\alpha$. Hence, a basis for $\mathcal L(\alpha P_{\infty})$ consists of $x^i$, $i \in \{ 0, \dots, \alpha \}$.  
This setting gives rise to the Optimal Polynomial Intersection (OPI) DQI. 
In later work \cite{Jordan_Gu_AG_DQI_25}, DQI is considered using the Hermitian curve. The Hermitian curve $\mathcal H_q$ is defined by 
\[ x^{q+1}=y^q+y\]
over $\F_{q^2}$. It has $q^3+1$ $\F_{q^2}$-rational points, including a unique point at infinity $P_{\infty}$, making it a maximal curve. The one-point Hermitian code $C_{\mathcal L}(D,\alpha P_{\infty})$ has length $q^3$, allowing for up to $q^3$ constraints in (\ref{eq:system}). A basis for $\mathcal L \left( \alpha P_{\infty} \right)$ comprises monomials $x^iy^j$ in two variables satisfying $iq+j(q+1) \leq \alpha$. 

Jordan and Gu note the following advantages of Hermitian curves for DQI.
\begin{enumerate}
    \item smaller alphabet sizes compared with Reed--Solomon codes of the same length, which translates
to more efficient quantum implementations, since representing large field elements
with qubits is generally expensive;
\item insight into the broader applicability
of the framework, helping identify which structural properties of algebraic codes are
essential for quantum optimization advantages;
\item an interpretation of HOPI as an approximate list-recovery problem for Hermitian codes.
\end{enumerate}
In the next section, we extend these notions beyond the Hermitian curve to other curve families.

\subsection{Max-LINSAT and DQI} We first recall the definition of max-LINSAT.

\begin{definition}[Maximum Linear Satisfiability (max-LINSAT), satisfaction fraction] \label{def:maxLINSAT}
Let $A\in\F_q^{n\times k}$ and $F_1,\ldots,F_n\subseteq\F_q$.
The Maximum Linear Satisfiability problem is
\[ \max_{x\in \F_q^k} \left|\{i\in [n] : (Ax)_i\in F_i\}\right|. \]
The satisfaction fraction of an assignment $x \in \F_q^k$  is \( \frac{s}{n}\), where \[s:=\mid \left\{ i \in [n]: a_{i1} x_1 + \dots + a_{ik} x_k \in F_i \right\} \mid\] denotes the number of satisfied constraints $\left(Ax \right)_i \in F_i$.
\end{definition}
Notice that if each $F_i=\{b_i\}$, this reduces to finding an
assignment $\x\in\F_q^k$ that satisfies the maximum number of equations
in the linear system $Ax=b$.

Given the system (\ref{eq:system}), there is a natural connection to coding theory. 
Collect the coefficients of the constraints into a matrix $B:=[b_{ij}] \in \F_q^{n \times k}$. Notice that we may define \( C=\{B\x:\x \in \F_q^k \} \subseteq \F_q^n\); equivalently, we may consider $B^T \in \F_q^{k \times n}$ as a generator matrix of a linear code $C$ of length $n$, meaning the code consists of all vectors $\cv \in \F_q^{1 \times n}$ that satisfy the constraints given by a particular vector $\x \in \F_q^{k \times 1}$. The decoding task amounts to determining $\x$ from some received vector $\w \in \F_q^{1 \times n}$. More precisely, determine $\x \in  \F_q^{k \times 1}$ so that as many of the equations 
\(
b_{i1}x_1+\dots+b_{ik}x_k=w_i
\)
are satisfied as possible. In this way, we may see max-LINSAT as a nearest-neighbor decoding problem. In particular, given a candidate solution $\x \in \F_q^k$, we may consider the residual vector $
  r(\x) \coloneqq B\x - \w$.
The Hamming weight $\operatorname{wt}_H\big(r(\x)\big)$, meaning the number of nonzero coordinates of $r(\x)$, is precisely the number of violated constraints.
Hence, maximizing the count of satisfied constraints is equivalent to
minimizing $\operatorname{wt}_H(r(\x))$. This is the {nearest (codeword) neighbor decoding} viewpoint \cite{BMT78}. In this way, for a received word $\w$, the optimization becomes
\[
  \min_{\x \in \F_q^k} \operatorname{wt}_H(B \x - \w) \quad \Longleftrightarrow \quad
  \min_{\y \in C} \operatorname{wt}_H(\y - \w).
\]
Choose \( \y^\star \in \arg\min_{y\in C} \operatorname{wt}_H(\y-\w) \). Any preimage \(\x^\star \in \mathbb F_q^k\) satisfying \(B\x^\star = \y^\star\) is an optimal assignment.

The DQI procedure associates with each constraint in (\ref{eq:system}) a function:  for $i \in [n]$ and  $F_i:=\left\{ c_{i1}, \dots, c_{ir} \right\}$, consider \( f_i: \mathbb{F}_{q}\) defined by \[
f_i(a)= \begin{cases} +1,&a\in F_i,\\ -1,&a\notin F_i. \end{cases}
\]
Then the optimization objective for max-LINSAT can be expressed as
\[
f(\mathbf{x}) = \sum_{i=1}^n f_i(\mathbf{b}_i \cdot \mathbf{x}),
\]  
the difference between the numbers of satisfied and unsatisfied constraints. DQI constructs quantum states that are preferentially weighted toward solutions with higher objective values. The algorithm prepares states of the form
\[
\ket{P(f)} \propto \sum_{\mathbf{x} \in \mathbb{F}_{p}^k} P(f(\mathbf{x})) \ket{\mathbf{x}},
\]  
where $P$ is a polynomial of degree $\ell$ chosen to amplify amplitudes corresponding to large $f(\mathbf{x})$; we will return to this choice of $\ell$ shortly. 
Here, the symbol $\propto$ indicates proportionality; that is, the state $|P(f)\rangle$ is obtained by normalizing the vector $\sum_{\mathbf{x}\in\mathbb{F}_q^k} P(f(\mathbf{x}))|\mathbf{x}\rangle$ so that its total probability is one, making it a valid quantum state. Explicitly, \[ |P(f)\rangle = \frac{1}{\sqrt{\sum_{\mathbf{x}\in\mathbb{F}_q^k}|P(f(\mathbf{x}))|^2}} \sum_{\mathbf{x}\in\mathbb{F}_q^k} P(f(\mathbf{x}))|\mathbf{x}\rangle. \]
Assignments for which $|P(f(x))|$ is larger have greater measurement probability, since the associated probability is proportional to $|P(f(x))|^2$.
Thus, the state \(|P(f)\rangle\) may be viewed as a weighted superposition of all possible assignments. Each basis state \(|\mathbf{x}\rangle\) appears with weight determined by the objective value \(f(\mathbf{x})\) through the polynomial \(P\). Consequently, assignments that satisfy more constraints are emphasized in the quantum state, making them more likely to be produced when the state is measured.
Sampling from this state in the computational basis then yields solutions biased toward optimal or near-optimal configurations. As described in \cite{ Jordan_2025_extended,jordan2025dqi_nature}, the construction of these states can be performed efficiently via a reduction to {syndrome decoding}. Specifically, the matrix $B$ defines a linear code
\(
{C} = \{ B \mathbf{x} : \mathbf{x} \in \mathbb{F}_{q}^k\}\),  
and the algorithm works with its dual code
\(
{C}^\perp = \{\mathbf{v} \in \F_{q}^n : B^T \mathbf{v} = \mathbf{0}\}\).  
    
    Preparing the quantum state $\ket{P(f)}$ requires evaluating amplitudes that can be expressed in terms of codewords of the dual code $C^\perp$. As shown in \cite{jordan2025dqi_nature}, this task reduces to solving syndrome decoding problems for $C^\perp$: 
       given a syndrome arising from error patterns of weight at most $\ell$, the decoder is used to recover the corresponding low-weight error vector whose support contributes to the amplitude of a basis state. To ensure that the recovered error vector is unique, DQI restricts attention to error patterns of weight at most \begin{equation} \label{eq:min_dist} \ell = \left\lfloor \frac{d^\perp-1}{2}\right\rfloor, \end{equation} where $d^\perp$ is the minimum distance of $C^\perp$. Since every error of weight at most $\ell$ is uniquely determined by its syndrome, the decoding procedure can be implemented coherently inside the quantum algorithm.

The polynomial $P$ is chosen by the DQI algorithm as it serves as an amplification
function that maps objective values $f(\mathbf{x})$ to amplitudes in the
quantum state. By selecting $P$ so that larger values of $f(\mathbf{x})$
receive disproportionately larger weights, DQI increases the probability of
observing high-quality solutions upon measurement. The effectiveness of this amplification depends on the degree
$\ell=\deg(P)$. 
Within the optimized polynomial family used by DQI, increasing the admissible degree can strengthen the amplification.
However, implementing
a polynomial of degree $\ell$ requires solving syndrome decoding problems
of weight up to $\ell$ in the dual code $C^\perp$. Consequently, the
maximum admissible degree is  limited by the decoding radius of the dual code $C^\perp$.

Larger values of $\ell$, up to the upper bound given in (\ref{eq:min_dist}), may provide stronger bias toward optima along with harder decoding problems. 
By \cite{jordan2025dqi_nature}   
    \[
\frac{s}{n} = \begin{cases}
\left( \sqrt{ \frac{\ell}{n}\left( 1-\frac{r}{q} \right)}+\sqrt{ \frac{r}{q} \left( 1-\frac{\ell}{n} \right)} \right)^2, &\textnormal{ if } \frac{r}{q} \leq 1-\frac{\ell}{n},\\
1, & \textnormal{ otherwise}
\end{cases} \]
the expected number $s$ of constraints satisfied by $\x \in \F_q^k$ sampled in the final measurement of DQI satisfies the semi-circle law.

\section{DQI with Codes from Curves}
\label{sec:main}

In this section, we will discuss the DQI framework with  algebraic geometry codes. Following results on codes from arbitrary curves over finite fields, we will investigate important curve families. This information will be used in the next section to draw comparisons among them.

For DQI with algebraic geometry codes, we require efficient decoding of the dual code $C^\perp$ through the error weights used by the algorithm. The current DQI framework relies on unique syndrome decoding: the decoding step must coherently recover a low-weight error vector from its syndrome in order to uncompute the corresponding error register \cite{Jordan_2025_extended,jordan2025dqi_nature}. Consequently, we restrict attention to efficient unique-decoding algorithms for algebraic geometry codes. Although stronger list-decoding results are known for AG codes (see \cite{beelen2024list, beelen2025faster,Beelen2022FastDO}), incorporating such algorithms into the DQI framework would require modifying the decoding step to account for multiple candidate error vectors and is beyond the scope of the present construction.

Lee, Bras-Amorós, and O'Sullivan \cite{6744599} give an interpolation-based decoder for general multipoint evaluation codes $C_{\mathcal L}(D,G)$ that uniquely decodes if $2\operatorname{wt}(e)<d_{\mathrm{LO}}$, where $d_{\mathrm{LO}}\geq n-\deg(G)$. Hence, efficient decoding occurs for
\[
\ell=\left\lfloor\frac{d_{\mathrm{LO}}-1}{2}\right\rfloor
\geq
\left\lfloor\frac{n-\deg(G)-1}{2}\right\rfloor.\]
For the dual code $C_{\mathcal L}(D,G)^\perp=C_{\Omega}(D,G),$
Sakata and Fujisawa \cite{8294201} give an efficient decoding algorithm based on the Berlekamp--Massey--Sakata algorithm and majority voting. Since the Goppa designed distance satisfies
\[d^\perp\geq d_{\mathrm{Goppa}}^\perp=\deg(G)-2g+2,\]
the DQI decoding parameter may be taken to satisfy
\[\ell\geq
\left\lfloor\frac{\deg(G)-2g+1}{2}\right\rfloor.\]
Their algorithm decodes through this Goppa-based radius and may decode further according to the Kirfel--Pellikaan bound. Thus, the DQI decoding requirement can be guaranteed without knowing the exact dual minimum distance. Analogous efficient decoding results are available for one-point AG codes \cite{179340,SAKATA1990207,476248}, while multipoint codes may also be decoded via embeddings into suitable one-point codes \cite{drake}. 
 
We now introduce COFI: Curve-based Optimal Function Intersection.

\begin{definition}
    [Curve-based Optimal Function Intersection (COFI)] 
    \label{def:COFI}
    Let $\mathcal X$ be a smooth projective curve over $\F_q$ with distinct rational points $P_1,\ldots,P_n$, $G$ be a divisor on $\mathcal X$ whose support contains none of the points $P_1, \dots, P_n$, and $F_1,\ldots,F_n \subseteq \F_q$. The \emph{Curve-based Optimal Function Intersection} (COFI) problem is to determine \[ \operatorname*{argmax}_{f\in \mathcal L(G)} \left| \left\{ i\in[n] : f(P_i)\in F_i \right\} \right|, \] that is, to find a function $f$ in the Riemann--Roch space $\mathcal L(G)$ of the divisor $G$ whose evaluations satisfy the maximum number of coordinatewise constraints. 
\end{definition}

Observe that when $\mathcal X=\mathbb{P}_{\F_q}^1$, the projective line over $\F_q$ and $G=\alpha P_{\infty}$ for some positive integer $\alpha$, 
the Riemann--Roch space $\mathcal L(G)$ consists of univariate polynomials of degree at most $\alpha$, recovering OPI. Likewise, taking $\mathcal X$ to be the Hermitian curve over $\F_{q^2}$ and $G=\alpha P_{\infty}$ recovers HOPI. A key feature of COFI is that the optimization is carried out over the full Riemann--Roch space $\mathcal L(G)$ and is therefore not restricted to polynomial functions. In general, elements of $\mathcal L(G)$ may be rational functions with prescribed poles at points in $\operatorname{supp}(G)$. This substantially broadens the class of functions that can be represented within the DQI framework compared with the previously studied OPI and HOPI formulations, which optimize over polynomial representations. In particular, COFI naturally accommodates multipoint divisors and Riemann--Roch spaces whose elements cannot, in general, be represented solely by polynomials. Thus, COFI provides a common framework for optimization problems arising from algebraic geometry codes on arbitrary curves. Its properties are given in the next result. We let $l(G,g):= \left \lfloor \frac{\deg(G) -2g + 1}{2}\right \rfloor$ throughout the remainder of the paper. 

\begin{theorem} \label{thm:main}
    Consider a smooth projective curve $\mathcal X$ of genus $g$ over a finite field $\F_q$ and a divisor $G$ on $\mathcal{X}$. 
Let \(P_1,\ldots,P_n \in \mathcal X(\F_q) \setminus  \supp(G)\), and assume \(2g-2 < \deg(G) <n\). Consider COFI over $\mathcal X$ with $n$ constraints, alphabet $\F_q$, and constraint sets of size $r$. Let $d^{\perp}$ denote the 
minimum distance of the dual of 
$C_{\mathcal{L}}(D,G)$ and assume efficient decoding up to $\ell = \left \lfloor \frac{d^{\perp}-1}{2}\right \rfloor$ errors. Then, DQI produces assignments whose  expected satisfaction fraction satisfies 
    \[
\frac{s}{n} \geq 
\left( \sqrt{
\frac{1}{n} \left( l(G,g) \right) \left( 1-\frac{r}{q} \right)
}+\sqrt{
\frac{r}{q} \left(  1 - \frac{1}{n} \left(l(G,g)\right)\right)
} \right)^2\]  if  $\frac{r}{q} \leq 1 - \frac{\ell}{n}\leq 1-\frac{l(G,g)}{n}$ and
$\frac{s}{n} = 1$
 otherwise. 
\end{theorem}

\begin{proof}
Let \(C=C_{\mathcal L}(D,G)\). By the standard duality theorem for algebraic geometry codes, \[ C^\perp = C_{\mathcal L}(D,D-G+K) \] for a canonical divisor $K$. If \(\deg (G) > 2g-2\), then the designed lower bound gives \[ d^\perp \ge \deg (G)-2g+2. \] Hence the dual unique decoding radius satisfies \[ \ell = \left\lfloor\frac{d^\perp-1}{2}\right\rfloor \ge \left\lfloor\frac{\deg (G)-2g+1}{2}\right\rfloor = l(G,g). \] Substituting this lower bound for \(\ell\) into DQI satisfaction formula yields the asserted lower bound. The second case follows from the saturation case of the semi-circle law.   
\end{proof}

\begin{corollary}[DQI from self-dual AG codes]
\label{cor:selfdual_dqi}
Let $\mathcal X$ be a smooth projective curve of genus $g$ over $\F_q$, let $D=P_1+\cdots+P_n,$ and let $G$ be a divisor such that  $C=C_{\mathcal L}(D,G)$ is self-dual. Assume $n$ is even. Then the satisfaction bound in Theorem~\ref{thm:main} may be evaluated using
\[
\frac{\ell}{n}
\geq
\frac{1}{n}
\left\lfloor\frac{n-2g}{4}\right\rfloor.
\]
\end{corollary}

\begin{proof}
Since $C$ is self-dual, $\dim(C)=\frac{n}{2}.$ Since $\deg(G)>2g-2$, the Riemann--Roch theorem gives $\dim(C)=\ell(G)=\deg(G)+1-g.$ Therefore, $\deg(G)+1-g=\frac{n}{2},$ and hence $\deg(G)=\frac{n+2g-2}{2}.$
Substituting this into the decoding radius
$$
\ell
\geq
\left\lfloor
\frac{\deg(G)-2g+1}{2}
\right\rfloor = l(G,g)
$$
gives
$$
\ell
\geq
\left\lfloor
\frac{n-2g}{4}
\right\rfloor.
$$

The result follows from Theorem~\ref{thm:main}.
\end{proof}

Next, we focus on algebraic geometry codes from two important curves with many rational points: the Suzuki curve and the extended norm--trace curve, which contains the Hermitian curve as a special case. The Suzuki curve, introduced in \cite{Suzuki1962}, is given by 
\[ 
y^q-y=x^{q_0}(x^q-x)
\]
over $\F_q$, where $q_0=2^t$ and $q=2q_0^2=2^{2t+1}$ for some positive integer $t$. 
Let $\mathcal S_q$ denote its smooth projective model.
Observe that $(a:b:1) \in \mathcal S_q (\F_q)$ since $a^q-a=0=b^q-b$ for all $a, b \in \F_q$. Moreover, $\mathcal S_q$ has a single point at infinity, $P_{\infty}:=(0:1:0)$. As a result, \[\mid \mathcal S_q \left( \F_q \right) \mid = q^2+1 .\] The genus of $\mathcal S_q$ is $g=q_0(q-1)$. It can be shown that $\mathcal S_q$ is an optimal curve \cite{vanderGeer1997HowTo}. Moreover, the dual of a one-point code on the Suzuki curve is another one-point code, namely for positive $\gamma \in \Z$, 
\[
C_{\mathcal L}(D, \gamma P_{\infty})^{\perp}=C_{\mathcal L}\left(D, \left(2q_0(q-1)+q^2-2-\gamma \right) P_{\infty} \right).
\]
These one-point codes are also explicit. Indeed, a basis for the Riemann-Roch space $\mathcal L(\gamma P_{\infty})$ is 
\[\left\{
x^{i}y^{j}v^{i'}w^{j'}\;: \begin{array}{l}
i,j,i',j'\in\N, 0\le j \le q-1,
0\le i' \le q_0-1,
0\le j' \le q_0-1,\\
iq + j(q+q_0)+i'(q+2q_0)+j'(q+2q_0+1)\le \gamma 
\end{array} \right\}\]
where $v:=y^{\frac{q}{q_0}}-x^{\frac{q}{q_0}+1}$ and $w:=y^{\frac{q}{q_0}}x^{\frac{q}{q_0^2}-1}+v^{\frac{q}{q_0}}$. 
Additional details on Suzuki codes may be found in \cite{Hansen1990}.

\begin{corollary}\label{Suzuki}
    Consider the Suzuki curve $\mathcal S_q$ over $\F_q$, and let $C^{\perp}$ be the dual of a one-point Suzuki code where $G = \alpha P_{\infty}$ with minimum distance $d_S^{\perp}$ and assume efficient decoding up to $\ell = \left \lfloor \frac{d_{S}^{\perp}-1}{2}\right \rfloor$ errors. Then DQI produces solutions with expected fraction of satisfied constraints
\[
\frac{s}{q^2} \geq \left( \sqrt{\frac{l(G,g)}{q^2}\left(1-\frac{r}{q}\right)}+\sqrt{\frac{r}{q} \left( 1-\frac{l(G,g)}{q^2}\right)} \right)^2
\]
if $\frac{r}{q} \leq 1-\frac{\ell}{q^2} \leq 1 - \frac{l(G,g)}{q^2}$ and $\frac{s}{q^2} =1$ otherwise. 
\end{corollary}
\begin{proof} Noting that \(n= q^2\) and \( \ell \geq  l(G,g)\) yields the result.
\end{proof}

Next, we consider  \Cref{thm:main} for a family of curves over $\F_{q^t}$ where $t \geq 2$ is an integer
that contains the Hermitian curve as a special case. Recall that the norm and trace of $a \in \F_{q^t}$ with respect to the extension field $\F_{q^t}$ of $\F_q$ are defined to be
\(
N(a):=a^{\frac{q^t-1}{q-1}}\) and \(
Tr(a):=\sum_{i=0}^{t-1} a^{q^i}\).
The extended norm--trace curve $\mathcal X_{q,t,u}$ over $\F_{q^t}$ is given by $x^u=Tr(y)$
where $u \mid \frac{q^t-1}{q-1}$. 
It has genus $g=\frac{(u-1)(q^{t-1}-1)}{2}$, and its number of $\F_{q^t}$-rational points is \[\mid \mathcal X_{q,t,u}(\F_{q^t}) \mid=q^{t-1}+u(q^t-q^{t-1})+1.\] Indeed, the number of $b \in \F_{q^t}$ with the same trace is $q^{t-1}$, and there are $(q-1)u+1$ elements  $a \in \F_{q^t} \setminus \{ 0 \}$ satisfying $a^u \in F_q$. Moreover, $\mathcal X_{q,t,u}$ has a unique point at infinity 
\[
P_{\infty}=\begin{cases}
    (0:1:0) & \textnormal{ if } u=\frac{q^t-1}{q-1}\\
    (1:0:0) & \textnormal{ if } u \neq \frac{q^t-1}{q-1}.
\end{cases}
\]
When $u=\frac{q^t-1}{q-1}$, $\mathcal X_{q,t,u}$ is referred to as the norm--trace curve and is given by 
\[N(x)=Tr(y).\]
The dual of $C_{\mathcal L}(D,\gamma P_{\infty})$ is 
\(C_{\mathcal L} \left(D,\gamma P_{\infty})^{\perp}=C_{\mathcal L}(D,(u(q^t-1)-1-\gamma) P_{\infty} \right)\).

\begin{corollary}\label{ent} Consider the extended norm--trace curve $\mathcal{X}_{q,t,u}$ over $\F_{q^t}$ where $u \mid \frac{q^{t}-1}{q-1}$. Let $C^{\perp}$ be the dual of an extended norm--trace code with minimum distance $d_{ENT}^{\perp}$ and assume efficient decoding up to $\ell = \left \lfloor \frac{d_{ENT}^{\perp}-1}{2}\right \rfloor$ errors. Then DQI produces assignments with expected fraction of satisfied constraints 
\begin{multline*}
\frac{s}{q^{t-1}+u(q^t-q^{t-1})} \geq \\
\left( \sqrt{\frac{l(G,g)}{q^{t-1}+u(q^t-q^{t-1})}\left(1-\frac{r}{q^t}\right)}+\sqrt{\frac{r}{q^t} \left( 1-\frac{l(G,g)}{q^{t-1}+u(q^t-q^{t-1})}\right)} \right)^2
\end{multline*}
if $\frac{r}{q^t} \leq 1-\frac{\ell}{q^{t-1}+u(q^t-q^{t-1})} \leq 1-\frac{l(G,g)}{q^{t-1}+u(q^t-q^{t-1})}$ and $\frac{s}{q^{t-1}+u(q^t-q^{t-1})} =1$ otherwise.

\end{corollary}
\begin{proof} Substituting \(q^{t-1} + u(q^t - q^{t-1})\) for \(n\) and \(\ell\geq l(G,g)\) yields the result.
\end{proof}

Taking $t=2$ and $u=q+1$ in Corollary \ref{ent} gives the Hermitian curve $\mathcal \mathcal H_q: x^{q+1}=y^q+y$ over $\F_{q^2}$ and recovers \cite[Theorem 1]{Jordan_Gu_AG_DQI_25} where $s/n$ becomes an equality since exact minimum distances of one-point Hermitian codes are known. Lastly, we consider \Cref{thm:main} for two-point Hermitian codes, where $H_q: x^{q+1} = y^q+y$ over $\F_{q^2}$ and $G = \alpha Q_1 + \beta Q_2$, $D = P_1 + \ldots + P_n$ and $n_{2pt} = q^3 - 1$. Here, $C_{\mathcal{L}}(D,\alpha Q_1 + \beta Q_2)^{\perp} = C_{\Omega}(D,\alpha Q_1 + \beta Q_2)$ and we will refer to $C_{\Omega}(D,\alpha Q_1 + \beta Q_2)$ as a two-point code as well in the remainder of the paper. We have the following result. 

\begin{corollary} Consider the Hermitian curve $H_q$ over $\F_{q^2}$. Let $C^{\perp}$ be the dual of a two-point Hermitian code $C_{\mathcal{L}}(D,\alpha Q_1 + \beta Q_2)$ with minimum distance $d_{2pt}^{\perp}$, and assume efficient decoding up to $\ell = \left \lfloor \frac{d_{2pt}^{\perp}-1}{2}\right \rfloor$ errors. Then DQI produces solutions with expected fraction of satisfied constraints 
\[\frac{s}{q^3-1} \geq \left( \sqrt{\frac{l(G,g)}{q^3-1}\left(1-\frac{r}{q^2}\right)}+\sqrt{\frac{r}{q^2} \left( 1-\frac{l(G,g)}{q^3-1}\right)} \right)^2\]

if $\frac{r}{q^2} \leq 1-\frac{\ell}{q^3-1} \leq 1 - \frac{l(G,g)}{q^3-1}$ and $\frac{s}{q^3-1} =1$ otherwise.

\end{corollary}
\begin{proof} Noting that \(q^3 - 1=n\) and $G=\alpha Q_1 + \beta Q_2$ yields the result.
\end{proof}

\subsection{Examples}

We now provide examples of the max-LINSAT formulation for the one-point Suzuki codes and two-point Hermitian codes.

\begin{example}
  Consider the Suzuki curve $\mathcal S_8: \mathcal S_8:y^8 - y = x^2(x^8 - x)$ over $\F_8$. Let $G=22P_{\infty}$. A basis for $\mathcal L(22P_\infty)$ is
\[
\mathcal B=\{1,x,y,v,w,x^2,xy,y^2,xw,yv\}
\]
where 
\(
v = y^{\frac{q}{q_0}}-x^{\frac{q}{q_0}+1} = y^4 - x^5\) and \( w = y^{\frac{q}{q_0}}x^{\frac{q}{q_0^2} - 1} - v^{\frac{q}{q_0}} = y^4x - v^4\) therefore $k=\ell(G) = 10$. Thus, the one-point Suzuki code
\(
C_{\mathcal L}(D,22P_\infty) \) with \( D=\sum_{a,b\in \F_8} P_{ab},
\)
has parameters $[64,10,\ge 42]$. Its generator matrix is 
\[
M=\big(f_i(P_{ab})\big)_{f_i\in\mathcal B,\ (a,b)\in \F_8^2}=
\begin{pmatrix}
1 & 1 & \cdots & 1\\
a_{1} & a_{2} & \cdots & a_{64}\\
b_{1} & b_{2} & \cdots & b_{64}\\
v(a_{1},b_{1}) & v(a_{2},b_{2}) & \cdots & v(a_{64},b_{64})\\
w(a_{1},b_{1}) & w(a_{2},b_{2}) & \cdots & w(a_{64},b_{64})\\
a_{1}^2 & a_{2}^2 & \cdots & a_{64}^2\\
a_{1}b_{1} & a_{2}b_{2} & \cdots & a_{64}b_{64}\\
b_{1}^2 & b_{2}^2 & \cdots & b_{64}^2\\
a_{1}w(a_{1},b_{1}) & a_{2}w(a_{2},b_{2}) & \cdots & a_{64}w(a_{64},b_{64})\\
b_{1}v(a_{1},b_{1}) & b_{2}v(a_{2},b_{2}) & \cdots & b_{64}v(a_{64},b_{64})
\end{pmatrix}.
\]
Any function \(f \in \mathcal L(22P_\infty)\) can be written as
\[
f
= c_0 + c_1 x + c_2 y + c_3 v + c_4 w + c_5 x^2 + c_6 xy + c_7 y^2 + c_8 xw + c_9 yv,
\]
with coefficient vector \(\mathbf{x} = (c_0,\dots,c_9)^T \in \F_8^{10}\).
Evaluating at a point \(P_{ab}\), we obtain \(f(P_{ab})
=\)
\[
c_0 + c_1 a + c_2 b + c_3 v(a,b) + c_4 w(a,b)
+ c_5 a^2 + c_6 ab + c_7 b^2
+ c_8 a\,w(a,b) + c_9 b\,v(a,b),
\]
which is linear in the coefficients \(\mathbf{x}\). Define the matrix \(B = M^T \in \F_8^{64\times 10}\). Then the \(i\)th row of \(B\) is
\[
\bigl(1,\ a_i,\ b_i,\ v(a_i,b_i),\ w(a_i,b_i),\ a_i^2,\ a_i b_i,\ b_i^2,\ a_i w(a_i,b_i),\ b_i v(a_i,b_i)\bigr).
\]
Thus, for \(\mathbf{x} \in \F_8^{10}\), \((B\mathbf{x})_i = f(P_{a_i,b_i});\)
that is, \(B\mathbf{x}\) is the evaluation vector of \(f\) at the affine rational points. Now, let \(F_i \subseteq \F_8\) be the set of allowed values at each point \(P_i\), for each $i \in [64]$. The Curve-based Optimal function intersection problem for the Suzuki curve over $\mathcal{F}_8$ can then be
formulated as the max-LINSAT problem
\[
\max_{\mathbf{x}\in \F_8^{10}}
\left|
\left\{
i \in [64] : (B\mathbf{x})_i \in F_i
\right\}
\right|,
\]
equivalently,
\[
\max_{f\in \mathcal L(22P_\infty)}
\left|
\left\{
 i \in [64] : f(P_i) \in F_i
\right\}
\right|.
\]
Hence, the variable \(\mathbf{x}\) represents the coefficient vector of
\(f \in \mathcal L(22P_\infty)\), while the matrix \(B = M^T\) encodes evaluation at
the affine rational points. In this way, the max-LINSAT problem is equivalent to
selecting a codeword in \(C\) that satisfies the maximum number of coordinatewise constraints.

\end{example}  

We now provide a max-LINSAT formulation for an example involving two-point Hermitian codes. This demonstrates that COFI uses rational functions that are not  polynomials. More details on these codes can be found in \cite[Example 5.3]{matthews2001weierstrass}. 

\begin{example} Let $\mathcal H$ denote the Hermitian curve over $\F_{64}$ given by $\mathcal H:y^8+y=x^9,$ and let $P_1=P_{00}, P_2=P_\infty.$
Then $\mathcal H$ has genus $g=\frac{q(q-1)}{2}=28,$ where $q=8$. Consider the divisor
\(
G=7P_{00}+82P_\infty\).
Since $\deg(G)=89>2g-2=54,$ $\ell(G)=\deg(G)-g+1=89-28+1=62.$ The divisors of the coordinate functions are
\(
(x)
=
\sum_{\beta^8+\beta=0}P_{0\beta}-8P_\infty\) and
\(
(y)=9(P_{00}-P_\infty)\).
By \cite[Theorem 3.4]{matthews2001weierstrass}, the two-point Weierstrass semigroup uses the rational functions $\frac{x^{q-s+1}}{y^{t-s+1}},$
whose pole divisors are
\[
\left(
\frac{x^{q-s+1}}{y^{t-s+1}}
\right)_\infty
=
\bigl((t-s)(q+1)+s\bigr)P_{00}
+
\bigl((q-t-1)(q+1)+s\bigr)P_\infty.
\]
For $q=8$, taking $t=s=r$ gives
\(
\left(
\frac{x^{9-r}}{y}
\right)_\infty
=
rP_{00}+(63-8r)P_\infty\).
Hence, for $r=1,\ldots,7$, we obtain the rational functions
\[
\frac{x^8}{y},\,
\frac{x^7}{y},\,
\frac{x^6}{y},\,
\frac{x^5}{y},\,
\frac{x^4}{y},\,
\frac{x^3}{y},\,
\frac{x^2}{y}  
 \in \mathcal L(7P_{00}+82P_\infty).\] It can be shown that the remaining basis elements can be taken from the basis \(
\mathcal B_\infty:=\{x^iy^j,
i\geq0,
0\leq j\leq7,
8i+9j\leq82\}\) of 
$\mathcal L(82P_\infty)$ which contains $55$ functions. Therefore, a basis for $\mathcal L(G)$ is
\[
\mathcal B
=
\mathcal B_\infty
\cup
\left\{
\frac{x^8}{y},
\frac{x^7}{y},
\frac{x^6}{y},
\frac{x^5}{y},
\frac{x^4}{y},
\frac{x^3}{y},
\frac{x^2}{y}
\right\},
\]
and $|\mathcal B|=55+7=62.$ Let $\mathcal B=\{f_1,\ldots,f_{62}\}.$
Equivalently, define \[ L_j= \begin{cases} 2, & j=-1,\\ 0, & 0\leq j\leq7, \end{cases} \qquad U_j= \begin{cases} 8, & j=-1,\\ \left\lfloor\dfrac{82-9j}{8}\right\rfloor, & 0\leq j\leq7. \end{cases} \] The basis can be written compactly as $\mathcal B = \left\{ x^iy^j: -1\leq j\leq7,\; L_j\leq i\leq U_j \right\}.$ Indeed, \( |\mathcal B| = 7+ \sum_{j=0}^{7} \left( \left\lfloor\frac{82-9j}{8}\right\rfloor+1 \right) = 7+55 = 62 = \ell(G)\),  so these functions form a basis of $\mathcal L(G)$. Consequently, every function $f\in\mathcal L(G)$ can be written uniquely as \( f(x,y) = \sum_{j=-1}^{7} \sum_{i=L_j}^{U_j} a_{ij}x^iy^j, \qquad a_{ij}\in\F_{64}.\) In particular, the terms corresponding to $j=-1$ are \[ \sum_{i=2}^{8}a_{i,-1}x^iy^{-1} = a_{2,-1}\frac{x^2}{y} +a_{3,-1}\frac{x^3}{y} +\cdots +a_{8,-1}\frac{x^8}{y}. \] Let $D = \sum_{P\in\mathcal H(\F_{64}) \setminus\{P_{00},P_\infty\}}P.$ Since the Hermitian curve has $q^3+1=8^3+1=513$ $\F_{64}$-rational points, we have $\deg(D)=511.$ The corresponding evaluation code is $C=C_{\mathcal L}(D,G).$ Since $\deg(G)<\deg(D)$, the evaluation map is injective, and therefore $\dim(C)=\ell(G)=62.$ Let $\mathcal B=\{f_1,\ldots,f_{62}\}$ be an ordering of the basis above, and let $\mathbf a=(a_1,\ldots,a_{62})^T\in\F_{64}^{62}$ be the corresponding coefficient vector, so that \[ f=\sum_{h=1}^{62}a_hf_h. \] Let $M= \bigl(f_h(P_m)\bigr)_{ f_h\in\mathcal B,\; P_m\in\operatorname{supp}(D)} \in\F_{64}^{62\times511}$ be a generator matrix for $C$, and define $B=M^T\in\F_{64}^{511\times62}.$ Then $(B\mathbf a)_m=f(P_m)$ for every $m\in[511]$. Now let $F_m\subseteq\F_{64}$, $m\in[511],$ be the allowed value set at the point $P_m$. The corresponding COFI problem can be written as \[ \max_{\mathbf a\in\F_{64}^{62}} \left| \left\{ m\in[511]:(B\mathbf a)_m\in F_m \right\} \right|. \] Equivalently, using the function representation above, the problem is \[ \max_{\substack{ a_{ij}\in\F_{64}\\ -1\leq j\leq7,\; L_j\leq i\leq U_j}} \left| \left\{ m\in[511]: \left( \sum_{j=-1}^{7} \sum_{i=L_j}^{U_j} a_{ij}x^iy^j \right)(P_m) \in F_m \right\} \right|;\] equivalently, \[ \max_{f\in\mathcal L(7P_{00}+82P_\infty)} \left| \left\{ m\in[511]:f(P_m)\in F_m \right\} \right|. \] Thus, the variables in the COFI instance are the $62$ coefficients $a_{ij}$ specifying a rational function \[ f(x,y) = \sum_{j=-1}^{7} \sum_{i=L_j}^{U_j} a_{ij}x^iy^j \] on the Hermitian curve. In particular, unlike OPI, the admissible functions are not restricted to polynomials: the $j=-1$ portion of the basis consists of the rational functions \[ \frac{x^2}{y}, \frac{x^3}{y}, \ldots, \frac{x^8}{y}. \]
\end{example}

We note that the dual code used for decoding in DQI is $C^\perp=C_\Omega(D,G).$
For the choice $G=7P_{00}+82P_\infty,$ \cite{matthews2001weierstrass} shows that $C^\perp = C_{\Omega}(D, 7P_{00} + 82P_{\infty})$
has parameters $[511,449,\geq38].$ 

We conclude this section by illustrating the satisfaction fraction advantages obtained from Suzuki and extended norm--trace codes, thereby demonstrating Corollary \ref{Suzuki} and Corollary \ref{ent}.

\begin{example}

For the DQI computation, we parameterize the dual code being decoded in this example as $C^\perp=C_{\mathcal L}(D,\gamma P_\infty), \gamma\in H(P_\infty).$ For $q=8$ and $n=64$, this code has parameters $[64,\ell(\gamma P_\infty),d^\perp], d^\perp\geq 64-\gamma,$ where exact values of $d^\perp$ are used when known. The corresponding decoding radius is \[ \ell=\left\lfloor\frac{d^\perp-1}{2}\right\rfloor. \] Recall that $C_{\mathcal L}(D,\alpha P_\infty)^\perp = C_{\mathcal L}\left( D, \left(2q_0(q-1)+q^2-2-\alpha\right)P_\infty \right).$ Thus, the parameter $\gamma$ in Table~\ref{tab:suzuki_q8_threshold} refers directly to the one-point Suzuki code $C^\perp$ used in the DQI decoding step.

\end{example}

\small
\setlength{\tabcolsep}{2.5pt}
\setlength{\LTleft}{0pt}
\setlength{\LTright}{0pt}

\begin{longtable}{c c c c c c c c c c}

\toprule
& & & \multicolumn{7}{c}{Lower bound on $s/n$} \\
\cmidrule(lr){4-10}
$\gamma$ & $d^\perp$ & $\ell$
& $r=1$ & $r=2$ & $r=3$ & $r=4$ & $r=5$ & $r=6$ & $r=7$ \\
\midrule
\endfirsthead

\toprule
& & & \multicolumn{7}{c}{Lower bound on $s/n$} \\
\cmidrule(lr){4-10}
$\gamma$ & $d$ & $\ell$
& $r=1$ & $r=2$ & $r=3$ & $r=4$ & $r=5$ & $r=6$ & $r=7$ \\
\midrule
\endhead

\midrule
\multicolumn{10}{r}{\emph{Continued on next page}}\\
\midrule
\endfoot

\endlastfoot

\endlastfoot

27  & $38^*$& $18$ & 0.633326 & 0.779998 & 0.880645 & 0.949609 & 0.990020 & 1.000000 & 1.000000 \\
28      & 36    & $17$ & 0.616353 & 0.765306 & 0.869047 & 0.941665 & 0.986234 & 1.000000 & 1.000000 \\
29      & 35    & $17$ & 0.616353 & 0.765306 & 0.869047 & 0.941665 & 0.986234 & 1.000000 & 1.000000 \\
30      & 34    & $16$ & 0.598911 & 0.750000 & 0.856763 & 0.933013 & 0.981763 & 1.000000 & 1.000000 \\
31      & 33    & $16$ & 0.598911 & 0.750000 & 0.856763 & 0.933013 & 0.981763 & 1.000000 & 1.000000 \\
32      & 32    & $15$ & 0.580971 & 0.734042 & 0.843750 & 0.923608 & 0.976562 & 0.999667 & 1.000000 \\
33      & 31    & $15$ & 0.580971 & 0.734042 & 0.843750 & 0.923608 & 0.976562 & 0.999667 & 1.000000 \\
34      & 30    & $14$ & 0.562500 & 0.717389 & 0.829959 & 0.913399 & 0.970584 & 0.998639 & 1.000000 \\
35      & 29    & $14$ & 0.562500 & 0.717389 & 0.829959 & 0.913399 & 0.970584 & 0.998639 & 1.000000 \\
36--37  & $28^*$& $13$ & 0.543457 & 0.699986 & 0.815331 & 0.902325 & 0.963768 & 0.996861 & 1.000000 \\
38      & 26    & $12$ & 0.523792 & 0.681770 & 0.799793 & 0.890312 & 0.956043 & 0.994270 & 1.000000 \\
39      & 25    & $12$ & 0.523792 & 0.681770 & 0.799793 & 0.890312 & 0.956043 & 0.994270 & 1.000000 \\
40--41  & $24^*$& $11$ & 0.503448 & 0.662664 & 0.783261 & 0.877272 & 0.947323 & 0.990789 & 1.000000 \\
42  & $22$    & $10$ & 0.482350 & 0.642572 & 0.765625 & 0.863092 & 0.937500 & 0.986322 & 1.000000 \\
43  & $21$    & $10$ & 0.482350 & 0.642572 & 0.765625 & 0.863092 & 0.937500 & 0.986322 & 1.000000 \\
44--45  & $20^*$& $9$  & 0.460407 & 0.621373 & 0.746752 & 0.847634 & 0.926439 & 0.980748 & 1.000000 \\
46      & 18    & $8$  & 0.437500 & 0.598911 & 0.726467 & 0.830719 & 0.913967 & 0.973911 & 1.000000 \\
47      & $17^*$& $8$  & 0.437500 & 0.598911 & 0.726467 & 0.830719 & 0.913967 & 0.973911 & 1.000000
\\
48--50  & $16^*$& $7$  & 0.413472 & 0.574982 & 0.704542 & 0.812109 & 0.899855 & 0.965607 & 0.999410 \\
51  & $13$    & $6$  & 0.388109 & 0.549305 & 0.680662 & 0.791481 & 0.883787 & 0.955555 & 0.997484 \\
52 & 12 & $5$ & 0.361103 & 0.521476 & 0.654378 & 0.768368 & 0.865315 & 0.943351 & 0.993915 \\
53--55  & $12^*$& $5$  & 0.361103 & 0.521476 & 0.654378 & 0.768368 & 0.865315 & 0.943351 & 0.993915 \\
56      & 8     & $3$  & 0.299965 & 0.456490 & 0.591378 & 0.711371 & 0.817940 & 0.909615 & 0.979653 \\
57--63  & $8^*$ & $3$  & 0.299965 & 0.456490 & 0.591378 & 0.711371 & 0.817940 & 0.909615 & 0.979653 \\
\multicolumn{10}{l}{\footnotesize{$^*$ denotes exact minimum distance}}\\
\bottomrule 
\caption{Lower bounds on $s/n$ for one-point Suzuki codes with $q=8$ and $n=64$, with exact values of minimum distance denoted by stars. For each $r=1,\ldots,7$, the lower bound is $\left( \sqrt{\frac{r}{8}\left(1-\frac{\ell}{64}\right)} + \sqrt{\frac{\ell}{64}\left(1-\frac{r}{8}\right)} \right)^2$ when $\frac{r}{8}\leq 1-\frac{\ell}{64}$, and is equal to $1$ otherwise.}\label{tab:suzuki_q8_threshold}
\end{longtable}

\begin{example}
We now consider DQI using extended norm--trace codes. As in the Suzuki case, we parameterize the dual code being decoded as $C^\perp=C_{\mathcal L}(D,\gamma P_\infty).$ Recall that \[ C_{\mathcal L}(D,\alpha P_\infty)^\perp = C_{\mathcal L}\left( D, \left(u(q^t-1)-1-\alpha\right)P_\infty \right). \] Thus, the parameter $\gamma$ below refers directly to the one-point extended norm--trace code $C^\perp$ used in the DQI decoding step. Table \ref{tab:ent_q8_t3_threshold} lists parameter choices for which the lower bound on the DQI satisfaction fraction exceeds \(0.933\), with \(t=3\), \(q=8\), and \(r/8^3=1/2\), so that \(r=256\). Here, $\frac{8^3 -1 }{8-1} = 73$, so $u =1$ or $u=73$. 
\end{example}

\small
\setlength{\tabcolsep}{4pt}
\setlength{\LTleft}{\fill}
\setlength{\LTright}{\fill}

\begin{longtable}{c c c c c}

\toprule
$u$ & $\gamma$ & $d^\perp$ & $\ell/n$ & $s/n$ \\
\midrule
\endfirsthead

\toprule
$u$ & $\gamma$ & $d^\perp$ & $\ell/n$ & $s/n$ \\
\midrule
\endhead

\midrule
\multicolumn{5}{r}{\emph{Continued on next page}}\\
\midrule
\endfoot

\endlastfoot

1  & 64   & 448 & 0.4355 & 0.9958 \\
1  & 96   & 416 & 0.4043 & 0.9908 \\
1  & 128  & 384 & 0.3730 & 0.9836 \\
1  & 160  & 352 & 0.3418 & 0.9743 \\
1  & 192  & 320 & 0.3105 & 0.9627 \\
1  & 224  & 288 & 0.2793 & 0.9487 \\
1  & 251  & 261 & 0.2539 & 0.9352 \\
\midrule

73 & 6144  & 26624 & 0.4062 & 0.9911 \\
73 & 8192  & 24576 & 0.3750 & 0.9841 \\
73 & 10240 & 22528 & 0.3437 & 0.9749 \\
73 & 12288 & 20480 & 0.3125 & 0.9635 \\
73 & 14336 & 18432 & 0.2812 & 0.9496 \\
73 & 16056 & 16712 & 0.2550 & 0.9358 \\
\bottomrule
\caption{Examples of the lower bound on $s/n$ for extended norm--trace codes with $q=8$, $t=3$, and $r/8^3=1/2$ (that is, $r=256$) for which $s/n>0.933$. Here $n=64+448u$, $d^\perp \ge n-\gamma$, and $\ell= \lfloor(d^\perp-1)/2\rfloor \geq \left \lfloor \frac{n-\gamma - 1}{2} \right \rfloor$.}\label{tab:ent_q8_t3_threshold}
\end{longtable}

For $q=8$, $t=3$, and $r/8^3=1/2$, the condition $s/n>0.933$ holds for a broad range of pole orders,  up to $\gamma/n<1/2$. Table~\ref{tab:ent_q8_t3_threshold} illustrates this for $u=1$ and $u=73$.

\section{Comparisons}\label{sec:families}

\subsection{Comparisons across Families of Curves}

In this section, we compare the DQI satisfaction fractions arising from several AG-code families with those arising from one-point Hermitian codes. We also provide examples for Suzuki codes, extended norm-trace codes, and two-point Hermitian codes. We now compare COFI using the extended norm--trace curve with \(t=2\) to HOPI using the Hermitian curve.

\begin{proposition}[Extended Norm-Trace Codes versus One-Point Hermitian Codes]\label{prop:ent_versus_herm} Let $C_H = C_{\mathcal{L}}(D_H,\alpha_HP_{\infty})$, $D_H = P_1 + \ldots + P_{q^3}$ be a one-point Hermitian code over $\F_{q^2}$ and $C_{ENT} = C_{\mathcal{L}}(D_{ENT}, \alpha_{ENT}P_{\infty})$, $D_{ENT} = P_1 + \ldots + P_{q+u(q^2-q)}$ be a one-point code over the extended norm--trace curve $y^q + y = x^u$ over $\F_{q^2}$, where $u \mid q+1$ and $u \neq q+1$. If $\alpha_{ENT}\ge 2g_{ENT}+2\left\lfloor \frac{n_{ENT}}{n_H}\ell_H\right\rfloor+1,$
then $\frac{\ell_{ENT}}{n_{ENT}}>\frac{\ell_H}{n_H}$ and the expected satisfaction fraction for COFI with the extended norm--trace code is greater than the HOPI satisfaction fraction obtained from the Hermitian code, meaning For a fixed $\rho$, this yields a strictly larger satisfaction fraction whenever both comparisons lie in the nonsaturated regime; if both are saturated, both fractions equal 1.
\end{proposition}

\begin{proof}
Notice that $\frac{\ell_{ENT}}{n_{ENT}}>\frac{\ell_H}{n_H}$ is equivalent to $\ell_{ENT}>\frac{n_{ENT}}{n_H}\ell_H.$ Since $\ell_{ENT}$ is an integer, it is enough to require $\ell_{ENT}\ge \left\lfloor \frac{n_{ENT}}{n_H}\ell_H\right\rfloor+1.$
Now $\ell_{ENT}= \left\lfloor \frac{d_{ENT}^\perp-1}{2}\right\rfloor
\ge
\left\lfloor \frac{\alpha_{ENT}+1-2g_{ENT}}{2}\right\rfloor.$
Therefore, it suffices that
\[
\left\lfloor \frac{\alpha_{ENT}+1-2g_{ENT}}{2}\right\rfloor
\ge
\left\lfloor \frac{n_{ENT}}{n_H}\ell_H\right\rfloor+1.
\]
A sufficient condition for this is $\alpha_{ENT}+1-2g_{ENT}\ge 2\left\lfloor \frac{n_{ENT}}{n_H}\ell_H\right\rfloor+2$ which is equivalent to $\alpha_{ENT}\ge 2g_{ENT}+2\left\lfloor \frac{n_{ENT}}{n_H}\ell_H\right\rfloor+1.$
Hence $\frac{\ell_{ENT}}{n_{ENT}}>\frac{\ell_H}{n_H}$ as desired.
\end{proof}

We now provide an example demonstrating Proposition \ref{prop:ent_versus_herm}. 

\begin{example} Let $q = 9$ and $u = 5$ so that $u \mid q+1$. Then, $g_H = \frac{9\cdot 8}{2} = 36, n_H = 9^3 = 729, g_{ENT} = \frac{8 \cdot 4}{2} = 16, n_{ENT} = 9(5\cdot8 +1) = 369.$ Let $\alpha_H = 100$. In this case, the exact dual minimum distance of the one-point Hermitian code attains the Goppa designed distance, therefore $d_H^\perp=\alpha_H-2g_H+2
=100-2(36)+2=30.$
Hence,
$$
\ell_H
=\left\lfloor\frac{d_H^\perp-1}{2}\right\rfloor
=\left\lfloor\frac{29}{2}\right\rfloor
=14.
$$

For the extended norm--trace code, the Goppa bound gives

$$
\ell_{ENT}\geq
\left\lfloor
\frac{\alpha_{ENT}-2g_{ENT}+1}{2}
\right\rfloor.
$$

Now choose $\alpha_{ENT} \ge 2g_{ENT}+2\left\lfloor \frac{n_{ENT}}{n_H}\ell_H\right\rfloor+1
= 32 + 2 \cdot 7 +1
= 47.$ Let $\alpha_{ENT} = 47$. Then, $\ell_{ENT} \geq 8.$ Thus, $\frac{\ell_H}{n_H} = \frac{14}{729} \approx 0.01920,
\quad
\frac{\ell_{ENT}}{n_{ENT}} \geq \frac{8}{369} \approx 0.02168.$ The resulting expected satisfaction fractions are

\small
\setlength{\tabcolsep}{4pt}
\setlength{\LTleft}{\fill}
\setlength{\LTright}{\fill}

\begin{longtable}{c c c}

\toprule
$r$ & $\frac{s_H}{n_H}$ & Lower bound on $\frac{s_{ENT}}{n_{ENT}}$\\
\midrule
\endfirsthead

\toprule
$r$ & $\frac{s_H}{n_H}$ & Lower bound on $\frac{s_{ENT}}{n_{ENT}}$ \\
\midrule
\endhead

\midrule
\multicolumn{3}{r}{\emph{Continued on next page}}\\
\midrule
\endfoot

\endlastfoot

1  & 0.061 & 0.066 \\
9  & 0.212 & 0.220 \\
18 & 0.347 & 0.355 \\
27 & 0.469 & 0.478 \\
36 & 0.583 & 0.592 \\
45 & 0.690 & 0.698 \\
54 & 0.790 & 0.797 \\
63 & 0.881 & 0.887 \\

\caption{Comparison of COFI with the extended norm--trace curve $y^9+y=x^5$ and HOPI with the Hermitian curve $y^9+y=x^{10}$, both over $\F_{81}$.}
\label{tab:ent_hermitian_q9}

\end{longtable}

\end{example}

We now have the following result, demonstrating that 
Suzuki-code DQI yields a larger satisfaction-fraction bound than Hermitian-code DQI at equal block length  ($n=q^6$), under the hypotheses below.

\begin{proposition}[Suzuki versus One-Point Hermitian Codes]\label{thm:suz_versus_herm_same_length}
Let $q = 2^{2m+1}$ with $m \ge 1$. Let $C_H = C_{\mathcal{L}}(D_H,G_H)$, $G_H = m_H P_{\infty}$ be a one-point Hermitian code over $\F_{q^4}$ and $C_S = C_{\mathcal{L}}(D_S,G_S)$, $G_S = m_S P_{\infty}$ a one-point Suzuki code over $\F_{q^3}$, with $m_H > 2g_H - 2$ and $m_S > 2g_S - 2$ and both scaled to the same length $n = q^6$. The Hermitian curve over $\F_{q^4}=\F_{(q^2)^2}$ has genus $g_H=\frac{q^2(q^2-1)}{2}.$ Moreover, since $q^3=2^{6m+3}=2^{2(3m+1)+1},$ the Suzuki curve over $\F_{q^3}$ has parameter $q_{0,S}=2^{3m+1}$
and genus $g_S=q_{0,S}(q^3-1)
=2^{3m+1}(q^3-1).$ For DQI, decoding is performed on the dual codes $C_H^\perp$ and $C_S^\perp$. The dual distance of $C_S^\perp$ satisfies $d_S^\perp \ge m_S + 2 - 2g_S,$ and the dual distance of $C_H^\perp$ is $d_H^{\perp}$ and for one-point Hermitian codes is known. The DQI parameters are
\[
\ell_H = \left\lfloor \frac{d_H^{\perp} - 1}{2}\right\rfloor ,
\qquad
\ell_S \geq \left\lfloor \frac{m_S - 2g_S + 1}{2} \right\rfloor.
\]

If $m_S \geq 2g_S + 1 + 2\ell_H$ then, $\frac{\ell_S}{n} > \frac{\ell_H}{n}$ and the expected satisfaction fraction for DQI is higher for the Suzuki code than for the Hermitian code for these choices.
\end{proposition}

\begin{proof}
Because both codes have identical length $n = q^6$, the inequality $\frac{\ell_S}{n} > \frac{\ell_H}{n}$
is equivalent to $\ell_S > \ell_H$. Using that $\ell_S \geq \left\lfloor \frac{m_S - 2g_S + 1}{2} \right\rfloor,$ $\ell_S > \ell_H$ holds whenever 
\[
\left \lfloor \frac{m_S - 2g_S + 1}{2}\right \rfloor > \ell_H \implies \frac{m_S - 2g_S + 1}{2} > \ell_H + 1.
\]
Solving for $m_S$ gives $m_S \geq 2g_S + 2\ell_H + 1,$
which proves the claim.
\end{proof}

\begin{remark} If $\ell_S = \left\lfloor \frac{m_S - 2g_S + 1}{2} \right\rfloor$, then the above statement is an if and only if. 
\end{remark}

We now have an example comparing the satisfaction fraction for Suzuki versus Hermitian over $q=8$, demonstrating Proposition \ref{thm:suz_versus_herm_same_length}.

\begin{example} Consider a one-point Hermitian code over $\F_{8^4}=\F_{4096}$ and a one-point Suzuki code over $\F_{8^3}=\F_{512}$. For the Hermitian curve, the corresponding Hermitian parameter is $64$, so $g_H=\frac{64(64-1)}{2}=2016, n_H=64^3=262144.$
For the Suzuki curve, we have $512=2^{2\cdot4+1}, q_0=2^4=16$,
and hence $g_S=q_0(512-1)=16(511)=8176,
n_S=512^2=262144.$ Thus, the two codes have the same length $n:=n_H=n_S=262144.$ Let $m_H=131000.$ For this one-point Hermitian code, $d_H^\perp=m_H-2g_H+2 =131000-2(2016)+2 =126970$
and
\[
\ell_H
=\left\lfloor\frac{d_H^\perp-1}{2}\right\rfloor
=\left\lfloor\frac{126969}{2}\right\rfloor
=63484.
\]
Now choose $m_S$ minimally so that $\ell_S>\ell_H$ gives $m_S=2g_S+2\ell_H+1=2(8176)+2(63484) + 1 = 143321$ and therefore \[\ell_S\geq \left\lfloor \frac{m_S-2g_S+1}{2}\right\rfloor
=\left\lfloor \frac{143321-2(8176)+1}{2}\right\rfloor
=63485.
\] Thus, $\frac{\ell_H}{n}=\frac{63484}{262144}\approx 0.2421722, \frac{\ell_S}{n}\geq \frac{63485}{262144}\approx 0.2421761$ and therefore $\frac{\ell_S}{n}>\frac{\ell_H}{n}.$ Now match $r$ by
\[
\frac{r_H}{q^4}=\frac{r_S}{q^3}
\quad\Longleftrightarrow\quad
\frac{r_H}{4096}=\frac{r_S}{512}
\quad\Longleftrightarrow\quad
r_H=8r_S.
\]

Let $\rho=\frac{r_S}{512}=\frac{r_H}{4096}.$
Using $\ell_S \geq 63485$, $\ell_H = 63484$ and 
\[
\frac{s}{n}
=
\left(
\sqrt{\frac{\ell}{n}(1-\rho)}
+
\sqrt{\rho\left(1-\frac{\ell}{n}\right)}
\right)^2,
\] we obtain 

\small
\setlength{\tabcolsep}{4pt}
\setlength{\LTleft}{\fill}
\setlength{\LTright}{\fill}

\begin{longtable}{c c c c}
\toprule
$r_S$ & $r_H$
& Lower bound on $\frac{s_S}{n}$
& $\frac{s_H}{n}$ \\
\midrule
\endfirsthead

\toprule
$r_S$ & $r_H$
& Lower bound on $\frac{s_S}{n}$
& $\frac{s_H}{n}$ \\
\midrule
\endhead
\midrule
\multicolumn{4}{r}{\emph{Continued on next page}}\\
\midrule
\endfoot
\endlastfoot
4 & 32 & 0.32163934 & 0.32163518 \\
32 & 256 & 0.48180245 & 0.48179800 \\
64 & 512 & 0.58999220 & 0.58998782 \\
128 & 1024 & 0.74209357 & 0.74208967 \\
192 & 1536 & 0.85034082 & 0.85033764 \\
256 & 2048 & 0.92840030 & 0.92839800 \\
320 & 2560 & 0.97925279 & 0.97925152 \\
\caption{Comparison of the DQI satisfaction bounds for the Suzuki and Hermitian codes at equal block length and matched constraint density $\rho=r_S/512=r_H/4096$.}
\label{tab}
\end{longtable}
Since $\ell_S/n>\ell_H/n$, the Suzuki code yields a strictly larger guaranteed lower bound on the expected fraction of satisfied constraints for each of the matched constraint densities shown above.
\end{example}

Lastly, we compare two-point Hermitian codes to one-point Hermitian codes in this setting. Let $\mathcal{X}/\F_{q^2}$ be the Hermitian curve and let $C=C_{\mathcal L}(D,G)$ be a two-point Hermitian code as above and suppose that $G$ satisfies
the Weierstrass-pair conditions of \cite[Theorem 4.1]{matthews2001weierstrass}.
Then $d(C^\perp)
=
d\bigl(C_{\Omega}(D,G)\bigr)
\ge
\deg(G)-2g+4.$ Under \cite[Theorem 4.4]{matthews2001weierstrass}, we have that $d(C^\perp)
\ge \deg(G)-2g+5.$ 

We now have the following result. 

\begin{proposition}[Two-point Hermitian Codes versus One-Point Hermitian Codes]

Let $C_{2pt} = C_{\Omega}(D, G)$, where $G = \alpha Q_{1}+ \beta Q_2$, be the dual of a two point Hermitian code, where $C_{2pt}$ satisfies \cite[Theorem 4.4]{matthews2001weierstrass}. Let $C_H$ be the dual of a one-point code Hermitian code of the same dimension. Suppose that
\[
\deg(G)=2g+q^2-aq-b-3,
3\leq a<b\leq q-1
\] as in \cite[Proposition 4.5]{matthews2001weierstrass}. Let $\ell_H = \left \lfloor \frac{d_H-1}{2}\right \rfloor$, $\ell_{2pt} = \left \lfloor \frac{d_{2pt}-1}{2}\right \rfloor.$ Then $\ell_{\mathrm{2pt}}\geq \ell_H+1$. Consequently, $\frac{\ell_{\mathrm{2pt}}}{n_{\mathrm{2pt}}}
>
\frac{\ell_H}{n_H},$
and hence the two-point Hermitian code yields a strictly larger expected DQI satisfaction fraction.

\end{proposition}

\begin{proof} By~\cite[Theorem~4.4]{matthews2001weierstrass}, the minimum distance of
the two-point code $C_{\Omega}(D, G)$, $G = \alpha P_1 + \beta P_2$ satisfies $d_{\mathrm{2pt}}
\geq
\deg(G)-2g+5.$
Using the assumption $\deg(G)=2g+q^2-aq-b-3$,
we obtain $d_{\mathrm{2pt}}
\geq
q^2-aq-b+2.$ For the corresponding one-point Hermitian code of the same dimension, the dual minimum distance is $d_H=q^2-aq-b.$
Hence, we have that $d_{\mathrm{2pt}}\geq d_H+2.$
Therefore,
\[
\ell_{\mathrm{2pt}}
=
\left\lfloor
\frac{d_{\mathrm{2pt}}-1}{2}
\right\rfloor
\geq
\left\lfloor
\frac{d_H+1}{2}
\right\rfloor.
\]
Since $d_H$ is an integer, $\left\lfloor
\frac{d_H+1}{2}
\right\rfloor
=
\left\lfloor
\frac{d_H-1}{2}
\right\rfloor+1
=
\ell_H+1.$
Thus, $\ell_{\mathrm{2pt}}\geq \ell_H+1.$ Finally, $n_{\mathrm{2pt}}=q^3-1<q^3=n_H.$
Combining gives $\frac{\ell_{\mathrm{2pt}}}{n_{\mathrm{2pt}}}
>
\frac{\ell_H}{n_H}$ and therefore the two-point Hermitian code gives a strictly larger expected satisfaction fraction than the one-point Hermitian code of the same dimension fraction.
\end{proof}

The following example uses the codes from \cite[Example 5.3]{matthews2001weierstrass}

\begin{example}
Consider the Hermitian curve $y^8+y=x^9$
over $\F_{64}$. Then $g=28$. Let $G=7P_1+82P_2$
so that $\deg(G)=89.$
The two-point code $C_{2pt} = C_{\Omega}(D,7P_1+82P_2)$ has parameters $[511,449,\geq 38].$
From \cite[Theorem 4.4]{matthews2001weierstrass}, we have that $d_{\mathrm{2pt}}
\geq
\deg(G)-2g+5
=
89-56+5
=
38.$
The one-point Hermitian code $C_H$ of the same dimension has parameters $[512,449,36]$.
Thus, $d_{\mathrm{2pt}}\geq 38>36=d_H$ and
\[
\ell_{\mathrm{2pt}}
\geq
\left\lfloor\frac{38-1}{2}\right\rfloor
=18,
\qquad
\ell_H
=
\left\lfloor\frac{36-1}{2}\right\rfloor
=17.
\]
Therefore, $\frac{\ell_{\mathrm{2pt}}}{n_{\mathrm{2pt}}}
\geq
\frac{18}{511}
\approx 0.03523,$
whereas $\frac{\ell_H}{n_H}
=
\frac{17}{512}
\approx 0.03320.$
Hence, $\frac{\ell_{\mathrm{2pt}}}{n_{\mathrm{2pt}}}
>
\frac{\ell_H}{n_H}.$
Thus, this two-point Hermitian code yields a larger relative decoding
radius than the one-point Hermitian code of the same dimension and,
therefore, a larger expected DQI satisfaction fraction for fixed
constraint density in the nonsaturated regime.
\end{example}

We now conclude this subsection by demonstrating with the following table that the alphabet size and length of the code affect the number of qubits required to represent one field element. We note that the extended norm--trace curve \(\mathcal{X}_{q,t,u, u = (q^t-1)/(q-1)}\) is the well known norm--trace curve.

{\scriptsize
\setlength{\tabcolsep}{2.1pt}
\setlength{\LTleft}{\fill}
\setlength{\LTright}{\fill}

\begin{longtable}{c c c c}

\toprule
Code & Alphabet size & Length $n$ & \# qubits per field element \\
\midrule
\endfirsthead

\toprule
Code & Alphabet size & Length $n$ & \# qubits per field element \\
\midrule
\endhead

\midrule
\multicolumn{4}{r}{\emph{Continued on next page}}\\
\midrule
\endfoot

\endlastfoot

Reed--Solomon
& $q$
& $q$
& $\left\lceil \log_2 q \right\rceil
= \left\lceil \log_2 n \right\rceil$
\\

Suzuki $\mathcal S_q$
& $q$
& $q^2$
& $\left\lceil \log_2 q \right\rceil
= \left\lceil \tfrac{1}{2}\log_2 n \right\rceil$
\\

Hermitian $\mathcal X_{q,2,q+1}$
& $q^2$
& $q^3$
& $\left\lceil \log_2 q^2 \right\rceil
= \left\lceil \tfrac{2}{3}\log_2 n \right\rceil$
\\

Extended norm--trace $\mathcal X_{q,t,(q^t-1)/(q-1)}$
& $q^t$
& $q^{2t-1}$
& $\left\lceil \log_2 q^t \right\rceil
= \left\lceil \tfrac{t}{2t-1}\log_2 n \right\rceil$
\\

Extended norm-trace ($t=2$) $\mathcal X_{q,2,u\mid q+1}$
& $q^2$
& $q+u(q^2-q)$
& $\left\lceil \frac{2\log_2q}{\log_2\!\bigl(q+u(q^2-q)\bigr)} \log_2n \right\rceil$
\\

\caption{Comparison of code length, alphabet size, and qubit cost per field element.}
\label{tab:qubit_cost}

\end{longtable}}

\subsection{Comparisons of DQI and Prange's Algorithm}

In this section, we provide some examples  comparing DQI performance with Prange's algorithm. We compare DQI against Prange's algorithm for Suzuki one-point codes over \(\mathbb F_8\), using exact distance data for the dual codes. We define $C^\perp = C_{\mathcal L}(D,\gamma P_\infty)$, so that, $\dim C^\perp = \ell(\gamma P_\infty),
k = \dim C = n - \ell(\gamma P_\infty).$
For the Suzuki curve over \(\mathbb F_8\), we have \(n = 64\). Using the exact dual distances \(d^\perp\) for certain values of $\gamma$ that were given in \cite{duursma_suzuki}, we set $\ell = \left\lfloor \frac{d^\perp - 1}{2} \right\rfloor.$ In the balanced case \(r = q/2 = 4\), we have \(\rho = 1/2\). The expected satisfaction fraction for Prange is
\[
\frac{s_{\mathrm{Prange}}}{n}
= \frac12 + \frac12 \frac{k}{n},
\]
while DQI achieves
\[
\frac{s_{\mathrm{DQI}}}{n}
=
\left(
\sqrt{\frac{\ell}{n}\cdot \frac12}
+
\sqrt{\frac12\left(1-\frac{\ell}{n}\right)}
\right)^2.
\]

Figure~\ref{fig:Suzuki-advantage} shows the comparison. Panel (a) plots the expected satisfaction fraction \( s/n\) for each value of \(\gamma\), while panel (b) shows the advantage
\[
\frac{ s_{\mathrm{DQI}} - s_{\mathrm{Prange}}}{n}.
\]
For all parameter choices shown in Figure \ref{fig:Suzuki-advantage}, DQI has a larger expected satisfaction fraction than Prange's algorithm.
The improvement is governed by the dual distance \(d^\perp\) and decoding radius. This demonstrates that even at small field size, Suzuki codes outperform the Prange baseline in expected satisfaction fraction.

\begin{figure}[h!]
\centering
\begin{tikzpicture}
\begin{groupplot}[
    group style={group size=1 by 2, vertical sep=2cm},
    width=0.4\textwidth,
    height=0.35\textwidth,
    grid=major,
    xlabel style={font=\small},
    ylabel style={font=\small},
    title style={font=\normalsize},
    tick label style={font=\small},
    legend style={font=\small, draw=gray!40, fill=white},
    xtick=data,
    xticklabel style={rotate=45, anchor=east},
]

\nextgroupplot[
    title={(a) Suzuki $q=8$},
    xlabel={Dual parameter $\gamma$},
    ylabel={$s/n$},
    ymin=0.60, ymax=1.01,
    legend pos=north east, x=.5cm,
    symbolic x coords={10,23,27,36,37,41,45,47,49,50,53,54,55,57,58,59,60,61,62,63},
]
\addplot[
    blue, thick, mark=*,
] coordinates {
    (10,0.993859)
    (23,0.963512)
    (27,0.949609)
    (36,0.902325)
    (37,0.902325)
    (41,0.877272)
    (45,0.847634)
    (47,0.830719)
    (49,0.812109)
    (50,0.812109)
    (53,0.768368)
    (54,0.768368)
    (55,0.768368)
    (57,0.711371)
    (58,0.711371)
    (59,0.711371)
    (60,0.711371)
    (61,0.711371)
    (62,0.711371)
    (63,0.711371)
};
\addlegendentry{DQI}

\addplot[
    orange, thick, mark=*,
] coordinates {
    (10,0.976562)
    (23,0.914062)
    (27,0.890625)
    (36,0.820312)
    (37,0.812500)
    (41,0.781250)
    (45,0.750000)
    (47,0.734375)
    (49,0.718750)
    (50,0.710938)
    (53,0.687500)
    (54,0.679688)
    (55,0.671875)
    (57,0.656250)
    (58,0.648438)
    (59,0.640625)
    (60,0.632812)
    (61,0.625000)
    (62,0.617188)
    (63,0.609375)
};
\addlegendentry{Prange}

\nextgroupplot[
    title={(b) Advantage},
    xlabel={Dual parameter $\gamma$},
    ylabel={$s_{\mathrm{DQI}}/n - s_{\mathrm{Prange}}/n$},
    ymin=0, ymax=0.23,
    x=0.5cm,
    symbolic x coords={10,23,27,36,37,41,45,47,49,50,53,54,55,57,58,59,60,61,62,63},
]

\addplot[
    blue, thick, mark=*,
] coordinates {
    (10,0.017296)
    (23,0.049450)
    (27,0.058984)
    (36,0.082012)
    (37,0.089825)
    (41,0.096022)
    (45,0.097634)
    (47,0.096344)
    (49,0.093359)
    (50,0.101172)
    (53,0.080868)
    (54,0.088681)
    (55,0.096493)
    (57,0.055121)
    (58,0.062934)
    (59,0.070746)
    (60,0.078559)
    (61,0.086371)
    (62,0.094184)
    (63,0.101996)
};

\end{groupplot}
\end{tikzpicture}
\caption{Comparison of DQI and Prange for Suzuki one-point codes over $\mathbb F_8$ in the balanced case $r=q/2=4$. Panel (a) plots the expected satisfaction fraction $s/n$, while panel (b) plots the advantage $s_{\mathrm{DQI}}/n-s_{\mathrm{Prange}}/n$.}
\label{fig:Suzuki-advantage}
\end{figure}
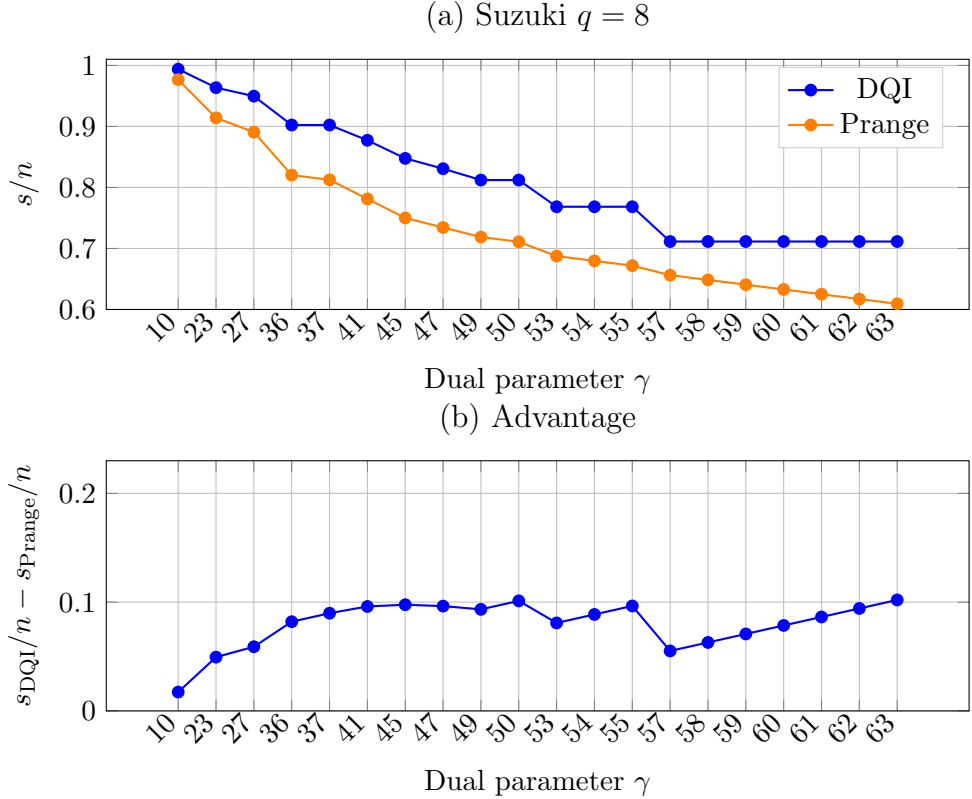

\section{Conclusion} \label{sec:conclusion}
In this paper, we investigated DQI performance with algebraic geometry codes by introducing COFI. The results in this paper demonstrate that families of AG codes can yield larger expected DQI satisfaction fractions than comparable Reed-Solomon constructions.  Also, Suzuki codes require fewer qubits per field element than Reed--Solomon codes of the same length. Our results show that, within the DQI framework, codes from extended norm--trace curves and Suzuki curves as well as a family of two-point Hermitian codes outperform one-point Hermitian codes by achieving larger satisfaction fractions under suitable assumptions on the divisor degree. Our work demonstrates that the DQI framework extends beyond Reed--Solomon and one-point Hermitian codes to broader families of algebraic geometry codes, and shows how the algebraic structure of these code families can be leveraged to improve DQI performance guarantees. 

An open direction is to explore DQI performance across broader families of codes. In particular, it may not be necessary to restrict to algebraic geometry codes, provided the chosen family fits within the DQI framework. We therefore aim to extend these comparisons to other code families and investigate whether DQI can be applied beyond the AG setting. Another interesting direction would be to understand how list-decoding algorithms for algebraic geometry codes can be incorporated into the DQI framework and, more generally, whether DQI can effectively exploit list-decoding methods beyond the unique-decoding regime. Such an extension could potentially allow DQI to operate with larger decoding radii and thereby improve the resulting optimization guarantees. Another natural direction is to investigate DQI using additional families of classical codes with efficient decoding algorithms, with the goal of identifying further code families whose structural and decoding properties yield improved quantum optimization performance.

\bibliographystyle{plain}
\bibliography{bib2}

\end{document}